\documentclass[12pt,english,twoside]{article}

\usepackage[a4paper, margin=1in]{geometry}
\usepackage{setspace}
\usepackage{fancyhdr}

\usepackage{amsmath, amssymb, amsthm, mathtools}

\usepackage[T1]{fontenc}
\usepackage{newtxtext}
\usepackage{newtxmath}
\usepackage{microtype}

\usepackage{graphicx}
\usepackage{subcaption}
\usepackage{booktabs}
\usepackage{float}

\usepackage{titlesec}
\usepackage{titlecaps}
\usepackage[hang, flushmargin]{footmisc}

\usepackage[style=american]{csquotes}
\usepackage[
authordate,
backend=biber,
sortcites=true,
dashed=false,
doi=true,
url=false,
isbn=false
]{biblatex-chicago}

\usepackage{xcolor}
\usepackage{hyperref}

\Addlcwords{a and amid as at but by down for from in into like near
	of off on onto over past since than the to under until up upon with}

\titleformat{\section}[hang]
{\filcenter\normalfont\bfseries}
{\thesection.}{0.5em}{\titlecap}

\titleformat{\subsection}[hang]
{\filcenter\normalfont\bfseries}
{\thesubsection.}{0.5em}{\MakeSentenceCase}

\titleformat{\subsubsection}[hang]
{\filcenter\normalfont\bfseries}
{\thesubsubsection.}{0.5em}{\MakeSentenceCase}

\renewenvironment{abstract}
{%
	\noindent\vspace{-20pt}\par            
	\setlength{\parindent}{0pt}%
	\vspace{10pt}\par
	\noindent\textbf{Abstract.}%
}
{%
	\par\vspace{30pt}%
}

\hypersetup{
	colorlinks=true,
	linkcolor=blue,
	citecolor=blue,
	urlcolor=blue,
	pdfauthor={Lars Nordmann},
	pdftitle={Newton's Second Law: A Structural Necessity of Mechanics
		Constrained by Excluded Perpetual Motion and Universal Free Fall}
}

\newcommand{\thmnotefont}{\mdseries}      
\newtheoremstyle{n2l}
{}{}                                  
{\normalfont}{}                       
{\bfseries}{.}{ }                     
{\thmname{#1}\thmnumber{ #2}\thmnote{ {\thmnotefont(#3)}}}

\theoremstyle{n2l}
\newtheorem{theorem}{Theorem}[section]    
\newtheorem{proposition}[theorem]{Proposition}
\newtheorem{definition}[theorem]{Definition}

\title{\LARGE{Newton's Second Law: \\A Structural Necessity of Mechanics Constrained by Excluded Perpetual Motion and Universal Free Fall}}
\author{\large{Lars Nordmann}\thanks{Email: lars.nordmann@rutgers.edu}}
\date{\normalsize November 2026}

\begin{document}

\maketitle

\begin{abstract}
For over three centuries, Newton's second law ($ F=ma $) has governed mechanics with extraordinary success, while its epistemic warrant has remained empirical. That merely empirical contingency vanishes once mechanics is constrained by two principles: the principle of excluded perpetual motion (PEPM) with its pre-Newtonian phenomenological content and the weak equivalence principle (WEP) with universal free fall as its primitive content. The PEPM secures the applicability of Suppes' measurement-theoretic framework to gravitational mass and force, exposing the formal structure latent in their empirical manifestations and establishing both as independent primitives. Gravity bridges statics and kinematics in Stevin's and Galileo's inclined-plane arguments, where the common sine factor in force and acceleration cancels in $F/a$. Cycle arguments under the PEPM extend this result to arbitrary force mechanisms, establishing $F/a$ as object-intrinsic and securing inertial mass as a well-defined physical quantity. Further constrained by universal free fall, the resulting quantitative structure compels the equivalence of inertial and gravitational mass, with Newton's second law as the operational expression of that equivalence. Direct implications include constraints on modified-inertia MOND that follow from the PEPM alone. The PEPM also places free-fall and torsion-balance experiments on equal footing as tests of the WEP. Together with the WEP, the derivation further grounds both the dimensional reduction of mechanics and the location-invariance of the Kibble-balance realization of the SI kilogram. Deeper still, Newton's second law derives its necessity from exclusion and invariance principles whose scope extends beyond mechanics itself.

\vspace{1.5em}
\end{abstract}

\newpage

\section{Introduction}\label{sec:introduction}
At the foundation of classical dynamics, Newton's second law \parencite{newton1687} embodies an enduring paradox: despite its extraordinary success across vast scales---from molecular to astronomical---its epistemic basis remains empirical. The law's pervasive application is therefore warranted by empirical adequacy rather than physical necessity. This epistemic gap has direct consequences.

It is precisely that lack of first-principles grounding that bifurcates admissible dynamics into two regimes: one that presupposes the law and one that permits its modification. The first encompasses all applications of Newtonian dynamics, though that presupposition is often tacit. For example, the Kibble-balance realization of the SI kilogram is location-invariant only if weight scales with local $g$ and if inertial and gravitational mass are equivalent \parencite{mana2022}, assumptions not ordinarily traced back to the law. The second regime encompasses modified-inertia formulations of MOND \parencite{milgrom1983}, which deliberately relax the universality of $ F=ma $ on the premise that Newton's second law is not compelled by deeper physical constraints.  

The same epistemic gap bears directly on the justification of experimental inference. In weak equivalence principle (WEP) testing, the WEP is probed directly through free-fall experiments such as MICROSCOPE and indirectly through mass-equivalence experiments such as E{\"{o}}t-Wash. Interpreting the latter as WEP tests remains contingent on assuming Newton's second law \parencites{wagner2012}[\S\,III]{dicasola2015}, leaving the two experimental routes on unequal epistemic footing.

This lack of first-principles grounding has kept Newton's second law under sustained foundational scrutiny. From the eighteenth century to the present, successive attempts to clarify its epistemic status have taken markedly different forms---rational demonstration, derivation from relativity principles, metaphysical grounding, methodological critique, conventionalist interpretation, operational reconstruction, and axiomatic formalization---yet the law's warrant remains unsettled. \textcite{euler1736}  claimed to demonstrate the law as a necessary truth, but his proof assumes that motions compose as forces do, which is the proportionality itself. Building on Laplace, nineteenth-century French treatises sought to derive the law from relativity principles; \textcite{darrigol2020} traces this tradition and shows that it rests on further dynamical assumptions. \textcite{kant1786} sought an \textit{a priori} grounding of the fundamental laws of mechanics, holding that empirical confirmation alone cannot establish their necessity; \textcite{mach1883} exposed circularities in the conceptual foundations of mechanics; \textcite{poincare1902} examined the extent to which the principles of mechanics express empirical content rather than convention; and \textcite{bridgman1927} pressed the demand that physical quantities be operationally grounded. Constructive approaches sought to resolve rather than diagnose the difficulty by defining inertial mass independently of Newton's second law: a quantity that could enter $F=ma$ without being defined by it. \textcite{jammer2000} traces this line through \textcite{saintvenant1845}, \textcite{mach1883}, and \textcite{weyl1949}. These constructions yield a well-defined mass only insofar as the mass ratios inferred across interactions are consistent, thereby relocating rather than resolving the question of why force and acceleration should be proportional. \textcite{suppes1951} supplied a rigorous measurement-theoretic framework for extensive quantities, while \textcite{mckinsey1953} axiomatized classical particle mechanics but retained mass and force as primitives and posited $F=ma$ as an axiom, thereby formalizing rather than grounding the relation---a limitation inherited by later axiomatizations, including \textcite[\S\,12.5]{suppes1957} and, in the continuum setting, \textcite{noll1959}. Later analyses have sharpened rather than resolved the underlying problem. \textcite{sneed1971} addressed the resulting dependence of mass and force on Newtonian mechanics by treating them as theoretical quantities determined within the theory as a whole rather than independently of it. \textcite{sklar2013}, in a modern examination of the foundations of dynamics, identifies the proportionality between the applied force and change of motion as the fundamental law of the Newtonian scheme, while emphasizing that foundational problems concerning mass and force persist. \textcite{fox2016} reconstructs Newton's recognition of inertial and gravitational mass as conceptually distinct quantities whose equivalence he provisionally established experimentally rather than deductively. Nor does special relativity close the foundational gap: in the Galilean limit ($ v \ll c $), it recovers $F=ma$ only in form, not in foundation.

The inference that doubling applied force doubles acceleration remains empirical, not deduced from first principles. Newton's second law therefore endures as a seemingly irreducible axiom---empirically warranted, yet retained by convention as a primitive rather than established as physically necessary.

This reconstruction is neither an axiomatization nor a reinterpretation, but a systematic derivation of Newton's second law from the principle of excluded perpetual motion (PEPM) with its pre-Newtonian phenomenological content and the WEP with universal free fall as its primitive content (see fig.~\ref{fig:fig1}). Any derivation of the second law must connect force to motion. Prior derivations did so by presupposing a quantitative relation between force and a body's response; the present reconstruction rests on an exclusion principle and an invariance principle that posit none. Under the PEPM, the applicability of Suppes' measurement framework \parencite{suppes1951} to gravitational mass and force is secured, thereby establishing them as independent primitives that formally ground the inclined-plane analyses of \textcite{stevin1586} and \textcite{galileo1638}. The PEPM is a premise Stevin and Galileo already accepted as a legitimate basis of mechanical reasoning \parencite{duhem1905}; its integration with a measurement-theoretic framework exposes the structure connecting their analyses. Traditionally treated as separate tracks, their synthesis bridges statics and kinematics, establishing inertial mass as an object-intrinsic quantity. Stevin establishes the force exerted by weight along an incline; Galileo establishes the corresponding acceleration. Their shared inclined-plane geometry supplies a common sine factor that cancels in $F/a$, leaving the ratio invariant under changes in force; the operational grounding of force extends this invariance beyond weight to arbitrary force mechanisms. When the formal structure thus established is further constrained by universal free fall, inertial and gravitational mass, though grounded in physically distinct phenomena, emerge as numerically equivalent. That equivalence, expressed in operational terms, is Newton's second law. The reconstruction does not merely replace one physical postulate with two others but grounds Newton's second law in principles that are epistemically prior to it: the PEPM is already a foundational constraint on Newtonian mechanics, while universal free fall is purely kinematic, invoking no concept of force.

The result shifts the law's epistemic status from an irreducible axiom of mechanics to a structural consequence of exclusion and invariance principles whose scope extends beyond mechanics itself. The derivation has direct consequences, two of which require only the PEPM: it constrains modified-inertia MOND and places free-fall and mass-equivalence tests of the WEP on equal epistemic footing. With the WEP added, it further grounds the location-invariance of the Kibble-balance realization of the SI kilogram and the dimensional reduction of mechanics.  

\medskip
\noindent Four background conventions and conditions apply throughout.

\smallskip
\noindent\textbf{Idealization.} All derivations in this reconstruction treat bodies as point masses and
assume idealized supporting elements (frictionless and massless pulleys, levers, and springs), thereby isolating the underlying mechanical structure.

\smallskip
\noindent\textbf{Frame of reference.} All quantities and relations are defined in instantaneous rest frames, where they remain uniquely specified; their expression in other reference frames follows from the corresponding transformations of force and acceleration.

\smallskip
\noindent\textbf{Dynamical response.} A free body's acceleration induced by a given resultant force is position-invariant: position enters the dynamical response only through the ambient field's contribution to the resultant force; the field's presence at that position has no independent effect on the induced acceleration.

\smallskip
\noindent\textbf{Symmetry of motion.} In the idealized mechanical configurations considered, every admissible motion is also admissible in reverse.
\medskip

\noindent None of these conditions presupposes Newton’s second law or any quantitative relation between force and acceleration.
	
\begin{figure}[H]
	\centering
	\includegraphics[width=0.86\textwidth]{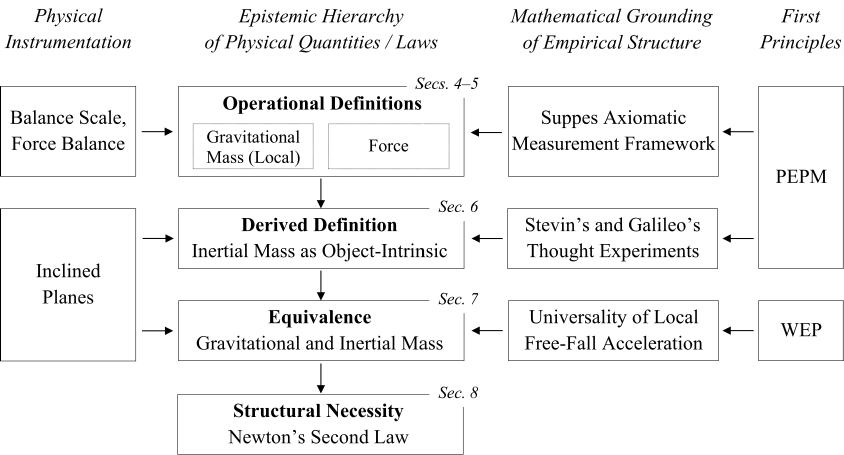}
	\caption{\small \textmd{Deductive reconstruction of Newton's second law.} Operational definitions of gravitational mass and force, admissible under the PEPM, ground the Stevin--Galileo synthesis, which establishes inertial mass as an object-intrinsic quantity. The WEP then compels the equivalence of inertial and gravitational mass; expressed in operational terms, that equivalence is $F=ma$.}\label{fig:fig1}
\end{figure}

\section{Physical First Principles}\label{sec:1stprinciples}
The derivation proceeds from two principles that constrain admissible dynamics by excluding unphysical behavior: the principle of excluded perpetual motion (PEPM) and the weak equivalence principle (WEP).

\medskip
\noindent\textbf{Principle of Excluded Perpetual Motion (PEPM).} No cyclic physical process can raise a weight indefinitely without external input (perpetual motion of the first kind). This is the pre-Newtonian phenomenological content of the PEPM \parencites[180--192]{banks2003}{mach1872}: an exclusion principle logically prior to, rather than derivative of, the conservation of mechanical energy. Reading the PEPM as derivative of the latter would invert this priority and render the reconstruction circular. Although expressed through the raising of a weight, the PEPM is not restricted to mechanical processes. The lift serves only as an operational witness of a net exploitable physical effect: any cyclic process capable of producing such a lift would violate the principle \parencite{helmholtz1847, bevilacqua1994}. 

\medskip
\noindent\textbf{Weak Equivalence Principle (WEP).} Locally, all bodies share the same free-fall acceleration, independent of mass or composition. This universality of free fall is the primitive content of the WEP \parencite{einstein1907}; reading it instead as the equivalence of inertial and gravitational mass would presuppose Newton's second law.

\section{Suppes' Axiomatic Framework for Measurement of Extensive Quantities}\label{sec:suppes}
Quantitative expression of physical order presupposes numerical representation of the corresponding quantities. Such representation requires a framework in which measurement rests on empirically realizable operations and preserves the quantities' empirical structure. \textcite{suppes1951} supplies precisely such a framework, formalizing measurement through comparison and concatenation under axioms that secure coherent representation. Beyond the operational grounding of measurement, the framework also yields the formal structure on which the Stevin--Galileo synthesis, and hence the entire derivation, depends.

\subsection{Suppes' representation theorem and measurement protocol}\label{subsec:suppesRTMP}

Suppes' framework comprises a representation theorem establishing the existence of a coherent numerical representation, and a measurement protocol that realizes it operationally.

\medskip
\noindent\textbf{Suppes' Representation Theorem.} 
Let $K$ be a nonempty set, $Q$ a comparison relation on $K$, and $\ast$ a concatenation operation on $K$. If $\langle K,Q,\ast\rangle$ satisfies Suppes' axioms, there exists a mapping $\mu:K \to \mathbb{R}^+$, unique up to a multiplicative scale factor, such that
\begin{align}
	\makebox[0.2\textwidth][l]{} &
	\makebox[0.4\textwidth][c]{$\mu(x\ast y) = \mu(x) + \mu(y)$} \label{eq:suppesadd} & \makebox[0.2\textwidth][l]{\text{(additive),}} \\
	\makebox[0.2\textwidth][l]{} &
	\makebox[0.4\textwidth][c]{$x \mathrel{Q} y \iff  \mu(x) \le \mu(y)$}\label{eq:suppesord} & \makebox[0.2\textwidth][l]{\text{(order-embedding).}}
\end{align}
The multiplicative scale freedom corresponds to the choice of an element $e \in K$ to serve as the unit representative, $\mu(e)=1$; once $e$ is fixed, $\mu$ is unique. The ordering $Q$ induces an equivalence relation on $K: x \sim y$ if and only if $x \mathrel{Q} y$ and  $y \mathrel{Q} x$. This relation partitions $K$ into $\mu$-indexed equivalence classes, allowing any member of a class to represent its common value.

\medskip
\noindent\textbf{Suppes' Measurement Protocol.}
Choose $e \in K$ as the unit representative. For $x \in K$ set $S_{[x][e]}:=\{m/n \mid m, n \in \mathbb{N}_{>0},\; nx \mathrel{Q} me\}$. Then $f_e(x) := \inf S_{[x][e]}$ is the unique additive, order-embedding numerical representation of $x$ guaranteed by the representation theorem. It admits operational realization through a systematically convergent measurement protocol: for each  $n \in \mathbb{N}_{>0}$ determine the least integer $m_x(n)$ such that $nx \mathrel{Q} m_x(n)e$. Then $\frac{m_x(n)}{n}-\frac{1}{n} < f_e(x) \le \frac{m_x(n)}{n}$, yielding increasingly precise bounds for $f_e(x)$ as $n \to \infty$. 

\subsection{Modified {S}uppes axioms}\label{subsec:modsuppes}
Although Suppes' \parencite*{suppes1951} axioms are formally sufficient, axiom V requires subtractive closure, a condition not \textit{a priori} warranted for physical quantities such as mass, whose realizations need not permit arbitrarily fine subdivision. Axiom V is therefore replaced by Connectedness and Right-Cancellability, with all other axioms unchanged. Table~\ref{tab:modsuppes} lists the modified axioms; proposition~\ref{prop:axiomVsufficiency} establishes their formal sufficiency.

\medskip
\begin{table}[h]
	\centering
	\begin{tabular}{@{}llll@{}}
		\toprule
		\multicolumn{2}{@{}l}{\textbf{Modified Suppes Axioms}} & & \textbf{Relation to \textcite{suppes1951}} \\
		\midrule
		M.1 & Closure under $\ast$: & $x \ast y \in K$                                              			& Suppes' A.II \\
		M.2 & Transitivity:         & $x \mathrel{Q} y$ and $y \mathrel{Q} z \implies x \mathrel{Q} z$ 			& Suppes' A.I \\
		M.3 & Associativity:        & $(x \ast y) \ast z \mathrel{Q} x \ast (y \ast z)$           		 		& Suppes' A.III \\
		M.4 & Compatibility:        & $x \mathrel{Q} y \implies x \ast z \mathrel{Q} z \ast y$     				& Suppes' A.IV \\
		M.5 & Connectedness:        & $x \mathrel{Q} y$ or $y \mathrel{Q} x$  				     				& replaces Suppes' A.V \\
		M.6 & Right-Cancellability: & $x \ast z \mathrel{Q} y \ast z \implies x \mathrel{Q} y$     				& replaces Suppes' A.V \\
		M.7 & Positivity:           & $\lnot(x \ast y \mathrel{Q} x)$                             				& Suppes' A.VI \\
		M.8 & Archimedean:          & If $x \mathrel{Q} y$ then $\exists n \in \mathbb{N}: y \mathrel{Q} nx$	& Suppes' A.VII \\
		\bottomrule
	\end{tabular}
	\caption{\small \textmd{Modified Suppes axioms.} Axiom V in \textcite{suppes1951} is replaced by \emph{Connectedness} and \emph{Right-Cancellability}; the remaining axioms are unchanged.}
	\label{tab:modsuppes}
\end{table}

\begin{proposition}[Formal Sufficiency of Modified Suppes Axioms]
	\label{prop:axiomVsufficiency}
	Replacing axiom V in \textcite{suppes1951} with Connectedness (M.5) and Right-Cancellability (M.6) preserves the formal sufficiency of the modified axiom set.
\end{proposition}
\begin{proof}
	In \textcite{suppes1951}, axiom V is invoked only in theorems 1, 5, 6, and 16. Each instance is subsumed by the modified axioms. Reflexivity in theorem 1 and connectedness in theorem 5 follow directly from M.5; right-cancellability in theorem 6 follows directly from M.6. Multiple-cancellation in theorem 16 follows by reduction: assume $nx \mathrel{Q} ny$. If $\lnot(x \mathrel{Q} y)$, then $y \mathrel{Q} x$ by M.5. Hence $n y \mathrel{Q} (n-1) y \ast x$ by M.4 and $n x \mathrel{Q} (n-1) y \ast x$ by M.2. Then $(n-1) x \mathrel{Q} (n-1) y$ by M.6, and thus $x \mathrel{Q} y$ by iterated reduction.
\end{proof}

\medskip
\noindent Accordingly, the modified axiom set guarantees a coherent numerical representation of empirically realizable comparison and concatenation operations---the formal basis for the operational definitions of gravitational mass and force.

\section{Operational Definition of Gravitational Mass}\label{sec:opdefmass}
Gravitational mass is the first primitive of the reconstruction and therefore requires operational definition. Suppes' framework supplies that definition while reproducing the exact numerical representation of conventional weighing. Conventional weighing is thus formalized, not revised.

\begin{definition}[Formalization of Gravitational Mass]
	\label{def:gravmass}
	Suppes' operations for gravitational mass are realized by the canonical operations of the balance scale: comparison $Q$ as weighing ``no more than,'' and concatenation $\ast$ as placing two objects jointly on the same pan.
\end{definition}
The corresponding measurement protocol assigns to an object $B$ at position $L$ the gravitational mass $m_L^G(B) = \inf\{m/n \mid m, n \in \mathbb{N}_{>0},\;  n B \mathrel{Q} m e_M \}\,\mathrm{kg}$, where $e_M$ denotes a conventional $1\,\mathrm{kg}$ mass transfer standard, such as one realized under the 2019 SI\@. Since the protocol is inherently local, gravitational mass remains position-indexed to leave open the possibility, though not the presumption, of such dependence.

\subsection{Admissibility of modified {S}uppes axioms}\label{subsec:admissibilitysuppesmass}
For Suppes' operational definition to be well-defined and to carry the requisite mathematical structure, the balance-scale operations must be shown to satisfy the modified axiom set. These axioms divide into two types: structural axioms (M.2--M.7) governing the domain's order-additive form, and domain-bounding axioms (M.1, M.8) constraining the class of admissible objects.

Structural axioms encode empirical behavior and must therefore be deduced from physical principles, whereas domain-bounding axioms must be justified by the scope of application. For gravitational mass, the Archimedean property (M.8) excludes objects of zero or infinite mass, exclusions coextensive with those of Newtonian dynamics, while closure (M.1) follows because the concatenation of two such objects again yields one of the same kind. Accordingly, proposition~\ref{prop:admissibilitysuppesmass} establishes the admissibility of the modified axiom set by deriving the structural axioms from the PEPM, with the domain-bounding axioms taken as given.

\begin{proposition}[Admissibility of Canonical Balance-Scale Operations] \label{prop:admissibilitysuppesmass}
	 Consider the canonical balance-scale operations at any fixed position $L$ on a domain satisfying closure (M.1) and the Archimedean property (M.8). Under the PEPM, these operations satisfy the modified axiom set.
\end{proposition}
\begin{proof}
	 M.1 and M.8 are satisfied by the stipulated domain. M.5 follows from the balance scale's determinate response: either pan descends or the scale remains in equilibrium. For M.2, suppose $x \mathrel{Q} y$, $y \mathrel{Q} z$ but $\lnot(x \mathrel{Q} z)$. Starting with $y$ and $z$ at higher and $x$ at lower elevations, consider the closed three-step cyclic exchange: $x$ versus $y$, $y$ versus $z$, and $z$ versus $x$. By $x \mathrel{Q} y$ and $y \mathrel{Q} z$, the first two exchanges proceed either reversibly or spontaneously, whereas the last proceeds strictly spontaneously since $\lnot(x \mathrel{Q} z)$. That last exchange can be coupled to raise a weight, thereby violating the PEPM\@. M.4 and M.6 follow by the same cycle argument. M.3 and M.7 require only two-step arguments: if M.3 were violated, mere regrouping of objects on a balance pan would permit a spontaneous cycle with net lift, violating the PEPM; if M.7 were violated, then $x \ast y \mathrel{Q} x$ would allow $y$ to be lifted in a spontaneous cycle, again violating the PEPM\@.	
\end{proof}

\medskip
\noindent The PEPM thus ensures that the conditions for Suppes' representation theorem are met. With a $1\,\mathrm{kg}$ mass transfer standard chosen as the unit representative, Suppes' protocol for the balance scale yields the unique additive, order-embedding numerical representation of gravitational mass. Since the operational definition is local in scope, gravitational mass remains provisionally indexed to position: an object judged heavier at one position need not be heavier at another. Positional invariance---and hence object-intrinsicness---is not assumed here; it emerges subsequently through equivalence with inertial mass under the WEP\@.

\subsection{Numerical coincidence with conventional weighing} \label{subsec:conv_suppes_coincidence}
Suppes' protocol is not an alternative to conventional weighing but a formal representation of that practice.

\begin{proposition}[Representational Equivalence]
	Consider the canonical balance-scale operations at any fixed position $L$ on a domain satisfying closure (M.1) and the Archimedean property (M.8). Under the PEPM, the conventional weighing protocol is well-defined and coincides numerically with Suppes' protocol.
\end{proposition}
\begin{proof}
	Conventional weighing assigns reference weights $r_{p/q}$ the denomination $p/q$ if and only if $q r_{p/q} \sim p e_M$, where $\sim$ denotes equilibrium. The gravitational mass of an object $B$ is determined by balancing $B$ against a multiset $R=\{r_{p_i/q_i}\}_{i=1}^{n}$ of reference weights and, whenever $B$ admits an exact balance, assigning to $B$ the denomination total $T(R)=\sum_{i=1}^{n}p_i/q_i$. On the set of reference weights, conventional weighing and Suppes' protocol produce identical numerical representations, since both are generated by the same equilibrium operations. Thus, $T(r_{p/q})=p/q=\mu(r_{p/q})$, where $\mu$ is Suppes' representation. Now let $B$ be balanced by a multiset $R= \{r_{p_i/q_i}\}_{i=1}^{n}$ of reference weights, so that $B \sim R$. The order-embedding property of Suppes' representation yields $\mu(B) = \mu(R)$, and its additivity yields $\mu(R) = \sum_{i=1}^{n}\mu(r_{p_i/q_i}) = \sum_{i=1}^{n}p_i/q_i$. Hence $T(R)=\mu(B)$. Since $\mu$ is well-defined, conventional weighing is likewise well-defined and coincides numerically with Suppes' protocol.
\end{proof}

\noindent Thus, whenever conventional weighing assigns an exact denomination total, it agrees with Suppes' representation. Where exact balance is unattainable because reference weights need not permit arbitrary subdivision, Suppes' protocol nevertheless determines $\mu(B)$ by bracketing $nB$ between consecutive integer multiples of the unit standard, yielding bounds of width $1/n$ that converge to $\mu(B)$. Suppes' protocol thus provides a completion of conventional weighing.

Figure~\ref{fig:fig2} illustrates the two operational protocols.

\medskip
\begin{figure}[H]
	\centering
	\begin{subfigure}[b]{0.4\textwidth}
		\centering
		\includegraphics[width=\textwidth]{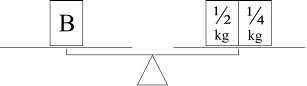}
		\captionsetup{justification=centering}
		\caption*{\small $m_L^G(B)=1/2 \, \mathrm{kg} + 1/4 \, \mathrm{kg} = 3/4 \, \mathrm{kg}$\\ \vspace{5pt} \small (a) Conventional protocol}
	\end{subfigure}
	\hfil
	\begin{subfigure}[b]{0.4\textwidth}
		\centering
		\includegraphics[width=\textwidth]{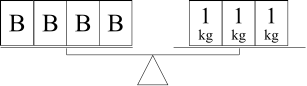}
		\captionsetup{justification=centering}
		\caption*{\small $m_L^G(B)=3/4 \, \mathrm{kg}$\\ \vspace{5pt} \small (b) Suppes' protocol}
	\end{subfigure}
	\caption{\small \textmd{Operational realization of gravitational mass.} (a) Conventional protocol: the gravitational mass of an object is determined by balancing it against a multiset of reference weights. (b) Suppes' protocol: gravitational mass is determined via comparison of multiple copies of the object against multiple copies of the $1\,\mathrm{kg}$ mass transfer standard. Under the PEPM, conventional weighing yields the same numerical representation as Suppes' protocol.}
	\label{fig:fig2}
\end{figure}

\noindent Conventional weighing is thereby established as well-defined, with its formal structure made explicit rather than merely presupposed. Positional invariance---and hence object-intrinsicness---is not yet claimed.

\section{Operational Definition of Force}\label{sec:opdef_froce}
Force is the second primitive and likewise requires operational definition, again supplied by Suppes' framework. 

A body's dynamical response is determined by the resultant force acting on it. Accordingly, resultants are hereafter referred to simply as forces. The direction of a force is given by its line of action. Its magnitude alone requires operational definition. The empirical content of that magnitude lies in a force's capacity to sustain or disrupt static equilibrium against another force acting along the same line. This content finds formal representation within Suppes' measurement framework, where empirical comparison and combination map directly onto the operations of comparison and concatenation. These operations are realized by redirecting the lines of action via fixed pulleys: antiparallel for comparison and parallel for concatenation.

\begin{definition}[Formalization of Force]
	\label{def:force}
	Suppes' operations for force magnitude are realized by the canonical operations of the pulley-coupled force-balance: comparison $Q$ as the relation ``no stronger than,'' and concatenation $\ast$ as combining two parallel forces.
\end{definition}

\noindent Under Suppes' measurement protocol, the magnitude of a force $A$ is $F(A) = \inf \{m/n \mid m, n \in \mathbb{N}_{>0},\; n A \mathrel{Q} m e_F \}$. Here, $e_F$ denotes a transferable force standard chosen by convention; the force unit is kept generic because force is introduced as a primitive.

\medskip
\noindent\textbf{Realization \textmd{(see fig.~\ref{fig:fig3})}.} A pulley-coupled force-balance comprises a taut rope passing over a fixed pulley that transmits opposing forces. When forces $F_1$  and $F_2$  act at the rope's opposite ends, the device realizes the comparison relation $Q$ through direct opposition, with equilibrium as reference. Codirected forces, when combined, realize the concatenation operation $\ast$. The force standard $e_F$ is realized by a reference spring.

\begin{figure}[H]
	\centering
	\begin{subfigure}[b]{0.3\textwidth}
		\centering
		\includegraphics[width=0.85\linewidth,trim={0pt 0pt 0pt 0pt},clip]{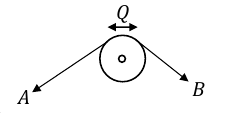}
		\captionsetup{justification=centering}
		\caption*{\small $A \mathrel{Q} B \iff F(A) \le F(B)$\\ \vspace{5pt} \small  (a) Force comparison $Q$}
	\end{subfigure}
	\hfil
	\begin{subfigure}[b]{0.3\textwidth}
		\centering
		\includegraphics[width=0.85\linewidth,trim={0pt 0pt 0pt 0pt},clip]{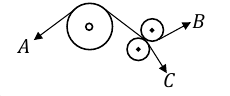}
		\captionsetup{justification=centering}
		\caption*{\small $B \ast C$\\ \vspace{5pt} \small (b) Force concatenation $\ast$}
	\end{subfigure}
	\vspace{-.3em}
	\caption{\small \textmd{Canonical operations of the pulley-coupled force-balance.} (a) Comparison $Q$: two forces are brought into direct opposition; the ``no stronger than'' relation is assessed relative to equilibrium. (b) Concatenation $\ast $: two forces $B$ and $C$ are codirected to form a joint force $B\ast C$. Together, these configurations realize the comparison and concatenation operations underlying Suppes' measurement protocol.}
	\label{fig:fig3}
\end{figure}

\subsection{Admissibility of modified {S}uppes axioms} \label{subsec:admissibilitysuppesforce}
As with gravitational mass, the operational definition of force requires the force-balance operations to satisfy the modified axiom set. The proof proceeds analogously. The Archimedean property (M.8) excludes forces of zero or infinite magnitude---exclusions already built into the Newtonian framework. Closure (M.1) holds because the concatenation of two such forces yields a force of the same kind. The structural axioms (M.2--M.7), by contrast, follow from the force-balance’s determinate response and the PEPM\@.

\begin{proposition}[Admissibility of Canonical Force-Balance Operations] \label{prop:admissibilitysuppesforce}
	Consider the canonical force-balance operations on a domain satisfying closure (M.1) and
	the Archimedean property (M.8). Under the PEPM, these operations satisfy the modified axiom set.
\end{proposition}
\begin{proof}
	M.1 and M.8 are satisfied by the stipulated domain. For the remaining axioms M.2–M.7, the argument parallels that in proposition~\ref{prop:admissibilitysuppesmass}. Here, nonzero resultants drive the displacement required for the PEPM argument, while reversible displacement in neutral equilibrium permits cycle closure.
\end{proof}

\medskip
\noindent The PEPM thus ensures that the conditions for Suppes' representation theorem are met. With a transferable force standard chosen as the unit representative, Suppes' protocol for the force-balance yields the unique additive, order-embedding numerical representation of force magnitude.

\subsection{The privileged role of weight}
\label{subsec:priviledgedroleweight}
The measurement of weight is operationally identical to that of gravitational mass, differing only in the choice of transferable standard---a force standard versus a mass standard. Weight therefore occupies a uniquely privileged role among physical forces: it both provides canonical representatives of force-value  equivalence classes and induces a formal correspondence between force and gravitational mass. 

\begin{proposition}[Formal Correspondence between Gravitational Mass and Weight]  \label{prop:relmassweight}
	Under the PEPM, the numerical representations of weight and gravitational mass are related by a similarity transformation:
	\begin{equation} \label{eq:relmassweight}
		W_L(B) = \gamma_L m_L^G(B),
	\end{equation}
	where $\gamma_L$ is a constant independent of the object $B$.
\end{proposition}
\begin{proof}
	For forces realized as weights, a balance scale functions simultaneously as a force-balance, with the canonical operations realized identically in both roles. The two roles thus induce isomorphic relational structures. The uniqueness clause of the representation theorem therefore guarantees that their numerical representations are related by a similarity transformation. Since the operational definition of gravitational mass is local, $\gamma_L$ remains position-indexed.
\end{proof}	

\section{Operational Definition of Inertial Mass} \label{sec:opdefinertialmass}
With force operationally defined via static equilibrium, characterizing its dynamical manifestation requires an independently grounded bridge between that definition and the induced kinematic response. Since spatial isotropy already compels the collinearity of force and acceleration, only the quantitative content of their relation remains to be established. The synthesis of two classically separate lines of analysis---Stevin's statics and Galileo's kinematics---supplies the deductive argument for that relation. Separately, they govern statics and kinematics; together, they establish $F/a$ as a strictly object-intrinsic quantity: the operational definition of inertial mass.

\subsection{Parallel force on inclined planes ({S}tevin)} \label{subsec:stevin}
\textcite{stevin1586} derived the sine law for parallel forces from his necklace argument, a  geometric manifestation of the PEPM\@. For a chain of identical beads draped over two planes at different inclinations (see fig.~\ref{fig:fig4}), the PEPM precludes spontaneous motion in either direction. Hence, static equilibrium requires the force acting along the incline to have magnitude 
\begin{equation}
	F^{||}_L(B)=W_L(B)\sin(\theta),
\end{equation}
where $\theta$ denotes the incline angle. 

\begin{figure}[H]
	\centering
	\begin{subfigure}[t]{0.2\textwidth}
		\centering
		\hspace{0.3em}\raisebox{2.5
			em}{\includegraphics[width=0.63\textwidth]{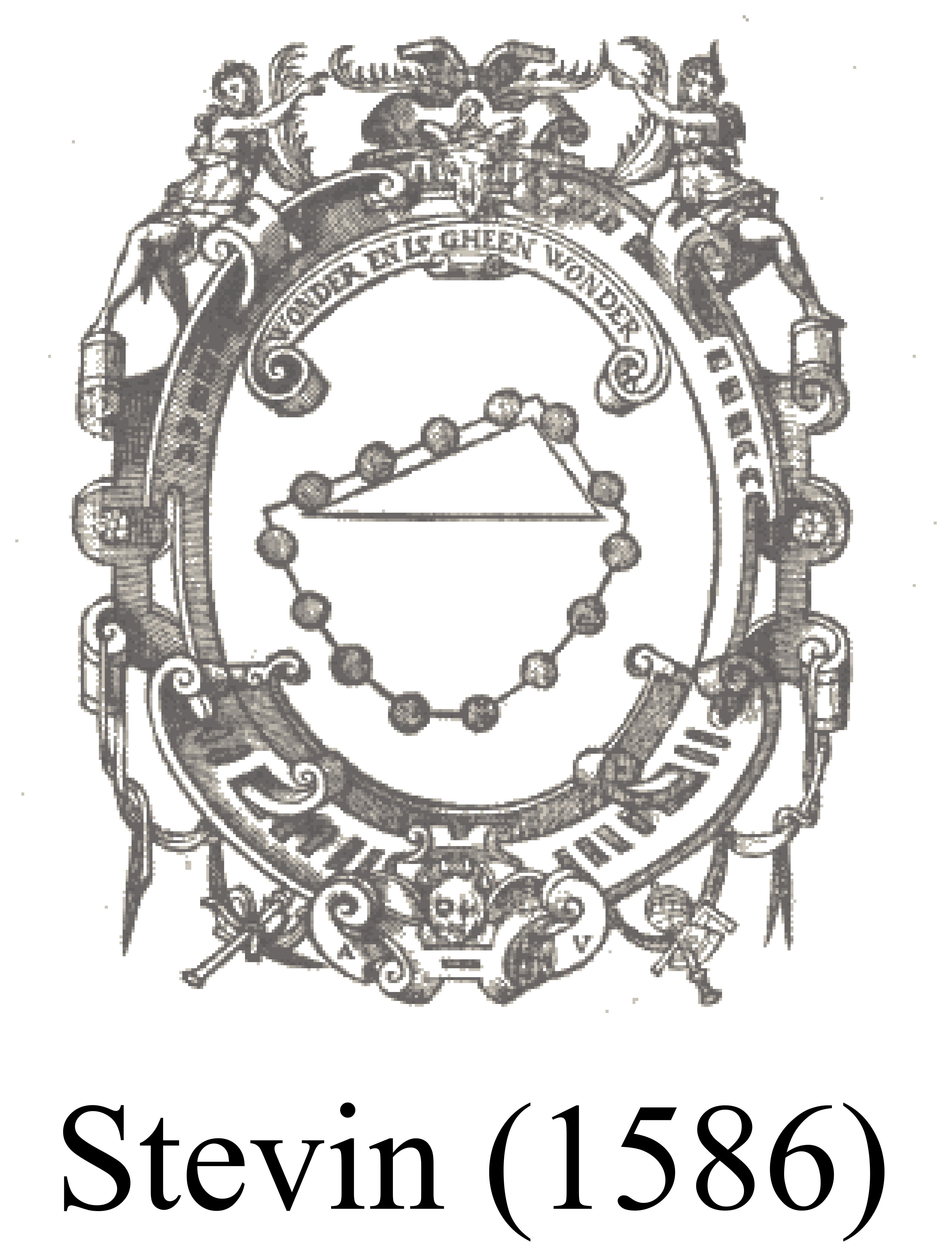}}
	\end{subfigure}
	\hfil
	\begin{subfigure}[t]{0.65\textwidth}
		\centering
		\hspace{-1em}\includegraphics[width=1\textwidth]{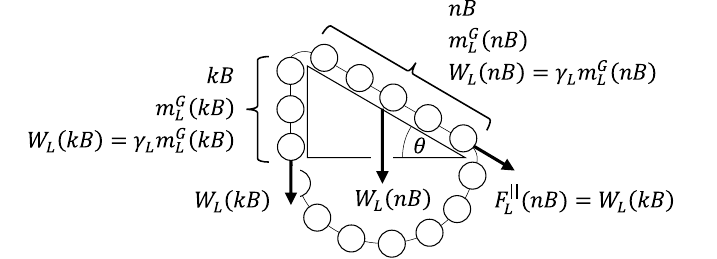}
	\end{subfigure}
	\vspace{0.3em}
	\caption{\small \textmd{Stevin's necklace argument.} Left: Stevin's original depiction \parencite*{stevin1586}. Right: modern representation of the same geometry, yielding $F^{||}_L(B) = W_L(B) \sin(\theta)$ under the PEPM\@.}
	\label{fig:fig4}
\end{figure}

\noindent Although the result may appear to follow directly from vector decomposition of forces, the sine law is deductively prior to that decomposition. 

\begin{proposition}[Stevin's Sine Law for Parallel Force] \label{prop:sinlawstevin}
	 Under the PEPM, $F_L^{||}(B) = W_L(B) \sin(\theta)$, where $F_L^{||}(B)$ denotes the component of $B$'s weight along an incline at angle $\theta$.
\end{proposition}
\begin{proof}
		For $\theta = \pi/2$, the result is immediate since the parallel force reduces to weight. For $\theta < \pi/2$, consider Stevin's necklace formed of identical copies of $B$ (see fig.~\ref{fig:fig4}). With $n$ copies on the incline and $k$ on the vertical, the necklace is materially realizable whenever $\sin(\theta)=k/n$; equilibrium then requires $F_L^{||}(nB)=W_L(kB)$, and additivity gives $F_L^{||}(B)=(k/n)\,W_L(B)=W_L(B)\sin(\theta)$. For arbitrary $\theta$, no such necklace need exist because of the mismatch between continuous geometry and discrete matter. That mismatch is resolved in the limit by a squeeze argument. The parallel force cannot be smaller on a steeper incline: if it were smaller at $\theta_2$ than at $\theta_1<\theta_2$, two copies of $B$ joined over the apex of a double incline would move spontaneously, and over a displacement $s$ the one on $\theta_2$ would rise by $s\sin(\theta_2)$ while the other descends only by $s\sin(\theta_1)$, violating the PEPM\@. Choose $k_n$ with $k_n/n \le \sin(\theta) < (k_n+1)/n$. The necklace inclinations with these sines bracket $\theta$, so monotonicity gives $(k_n/n)\,W_L(B) \le F_L^{||}(B) \le \bigl((k_n+1)/n\bigr)W_L(B)$. The bracketing collapses as $n \to \infty$, so $F_L^{||}(B) = W_L(B)\sin(\theta)$.
\end{proof}

\subsection{Acceleration on inclined planes ({G}alileo)} \label{subsec:galileo}
Galileo's sine law for inclined-plane acceleration \parencite{galileo1638} states:
\begin{equation}
	a_L(B) = g_L(B) \sin(\theta),
\end{equation}
where $a_L(B)$ and $g_L(B)$ denote the inclined-plane and free-fall accelerations of body $B$ at position $L$. The sine law is the kinematic expression of Galileo's principle of equal heights (PEH): a body descending from a given height acquires exactly the speed required to reascend to that height, independent of path (see fig.~\ref{fig:fig5}). The PEH itself follows from the PEPM under symmetry of motion, since a reascent above the initial height would permit perpetual gain, while one below would yield perpetual gain upon reversal.

\medskip
\begin{figure}[H]
	\centering
	\includegraphics[width=0.35\textwidth]{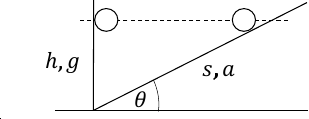}
	\vspace{5pt}
	\caption{\small \textmd{Galileo's principle of equal heights (PEH).} A body descending from a given vertical height acquires exactly the speed required to reascend to that height, independent of path.}
	\label{fig:fig5}
\end{figure}

\begin{proposition}[Galileo's Sine Law for Inclined-Plane Acceleration] \label{prop:sinlawgalileo}
	Under the PEPM and symmetry of motion, the instantaneous accelerations from rest satisfy $a_L(B) = g_L(B) \sin(\theta)$.
\end{proposition}
\begin{proof}
	Pure kinematics gives $a \mathrm{d}s=v \mathrm{d}v$. For locally uniform accelerations, integration from rest yields $v^2=2as$. Despite its formal resemblance to the work-energy relation, this result is pre-Newtonian in logical status and invokes neither force nor energy. Accordingly, $2sa_L(B) = v^2 = 2hg_L(B)$, where $v$ is the speed at the low point. Inclined-plane geometry gives $h/s = \sin(\theta)$, and hence $a_L(B)/g_L(B) = \sin(\theta)$. In the limit $s \to 0$, $a_L(B)$ is the instantaneous acceleration from rest, and the uniformity assumption lapses.
\end{proof}

\subsection{Inertial mass as an object-intrinsic quantity} \label{subsec:inertmassobjintrinsic}
The synthesis of Stevin's and Galileo's inclined-plane arguments provides the deductive grounding needed to establish $F/a$ as an object-intrinsic quantity, thereby giving precise physical content to the notion of a body's inertia.

\begin{proposition}[Structural Foundation for Inertial Mass] \label{prop:faratio}
	Assume the PEPM and symmetry of motion. For any body  $B$, any nonzero finite force-value $F$, and any position $L$ with nonzero gravitational field,
	\begin{equation} \label{eq:faratio}
		\frac{F}{a(B,F)} =\frac {\gamma_L m_L^G(B)}{g_L(B)},
	\end{equation}
	where $a(B,F)$ denotes the scalar acceleration of $B$ under $F$. Moreover, the ratio $F/a(B,F)$ is a well-defined, object-intrinsic quantity of $B$. 
\end{proposition}
\begin{proof}
	The proof proceeds in three steps. First, the induced acceleration is shown to be independent of the representative of a force-value. This permits realization of the force-value by weight, yielding equation~\ref{eq:faratio} for a suitable compound body $nB$. A reduction argument then yields equation~\ref{eq:faratio} for $B$, thereby establishing $F/a(B,F)$ as well-defined and object-intrinsic.
	\vspace{3pt}\\
	(i) Independence of representation: Suppose two representatives of $F$, $A \sim A'$, each exerting that force on $B$, induce different accelerations. A cycle comprising sequential, equally sized but opposite strokes by $A$ and $A'$ upon $B$ would then fail to return $B$ to rest, leaving residual motion. Because $A \sim A'$ entails static equilibrium, the two force agents can be coupled in opposition and reversibly reset to their initial configurations. Converting $B$'s residual motion into a lift completes the cycle with a net gain, violating the PEPM\@. Hence, $a(B,F)$, and therefore $F/a(B,F)$, are invariant under the choice of representative of $F$.
	\vspace{3pt}\\
	(ii) Invariance across force-value for $nB$: Fix $L$ and choose $n \in \mathbb{N}$ sufficiently large so that $F \le W_L(nB)$; the additivity of weight guarantees the existence of such $n$. As a representative of the force-value $F$, choose the parallel weight component of $nB$ at an incline $\theta$ such that $F = W_L(nB)\sin(\theta)$. By proposition~\ref{prop:sinlawgalileo}, $a(nB,F) = g_L(nB)\sin(\theta)$. Dividing and applying proposition~\ref{prop:relmassweight} yields $F/a(nB,F) =\gamma_L m_L^G (nB)/g_L(nB)$.
	\vspace{3pt}\\
	(iii) Reduction from $nB$ to $B$: If $n>1$, a closed-cycle argument establishes the kinematic scaling between $a(B,F)$ and $a(nB,F)$. Consider a smooth track with a horizontal segment transitioning into a vertical rise. On the horizontal segment, a spring operating over an infinitesimal stroke $\mathrm{d}s$ between two stops represents $F$ up to $O(\mathrm{d}s)$ along that stroke. Kinematics again gives $a \mathrm{d}s=v \mathrm{d}v$. Therefore, for motion from rest, $v^2=2a(B,F) \mathrm{d}s+o(\mathrm{d}s)$ along the horizontal spring stroke, and $v^2 = 2g_L(B)h$ along the vertical rise, which together yield $g_L(B)h=a(B,F)\mathrm{d}s+o(\mathrm{d}s)$. Let $h_1$ and $h_n$ denote the heights attained by $B$ and $nB$, respectively, under the same spring stroke $\mathrm{d}s$. Consider a closed cycle: $nB$ descends from $h_n$ and compresses the spring; $n-1$ copies of $B$ are removed; one copy ascends to $h_1$; all $n$ copies are then reversibly redistributed, for example by an idealized lever, to the common height $h_1/n$. The PEPM forbids net lifting in cycles, hence $h_1  \le nh_n$. Applied to the reversed cycle, it yields the reverse inequality, forcing $h_1  = nh_n$. In the limit $\mathrm{d}s \to 0$, this yields $a(B,F)/g_L(B) = n a(nB,F)/g_L(nB)$. Together with the result of (ii) and the additivity of gravitational mass, this yields $F/a(B,F) = \gamma_L m_L^G(B)/g_L(B)$. Its right-hand side depends at most on $B$ and $L$; under the dynamical-response condition, the left-hand side is independent of $L$. Hence $F/a(B,F)$ constitutes a well-defined, object-intrinsic quantity of $B$.
\end{proof}

\medskip
\noindent Accordingly, the object-intrinsic character of $F/a(B,F)$, ordinarily postulated or axiomatized, follows necessarily from the PEPM and the stated conditions. The derivation invokes a position with nonzero gravitational field. That requirement is existential, not universal: object-intrinsicness carries the ratio to positions where no field acts.

\medskip
\begin{definition}[Inertial Mass] \label{def:definertialmass}
Assume the PEPM\@. For any body $B$ and any nonzero force-value $F$, let $a(B,F)$ denote the scalar acceleration of $B$ under $F$. The inertial mass of $B$ is given by
\begin{equation} \label{eq:opdefinertialmass}
	m^I(B) \overset{\mathrm{def}}{=} F/a(B,F).
\end{equation}
\end{definition}

\medskip
\noindent Inertial mass is thereby operationally grounded without presupposing Newton's second law.

\section{Equivalence of Inertial and Gravitational Mass}
The PEPM alone establishes the operational structure defining gravitational mass, force, and inertial mass, while leaving the relation between the two mass concepts undetermined.

\begin{proposition}[WEP and Mass Equivalence]\label{prop:massequivalence}
	 Under the PEPM, the WEP is equivalent to universal proportionality between gravitational and inertial mass: $m_L^G(B)=c\,m^I(B)$, where $c$ is independent of both $B$ and $L$. With a suitable choice of force unit, $c$ may be set to unity.
\end{proposition}
\begin{proof}
	Equation \ref{eq:faratio} with definition \ref{def:definertialmass} gives $\frac{m_L^G(B)}{m^I(B)} = \frac{g_L(B)}{\gamma_L}$. Since $\gamma_L$ is independent of $B$, the WEP is equivalent to $m_L^G(B)/m^I(B)$ being independent of $B$. Under either equivalent condition, $\frac{m_L^G(B)}{m^I(B)} = \frac{m_L^G(e_M)}{m^I(e_M)} =:c$. Since $m_L^G(e_M)=1\,\mathrm{kg}$ by operational definition and $m^I(e_M)$ is position-invariant, $c$ is likewise independent of $L$, and hence of both $B$ and $L$, establishing the stated proportionality. The constant $c$ reconciles the independently chosen gravitational-mass and force standards. Rescaling the force unit by a factor $1/c$ rescales inertial-mass values by $c$, so that $c=1$.
\end{proof}

\medskip
\noindent Thus, under the PEPM, universal free-fall and mass-equivalence experiments are epistemically equivalent tests of the same underlying physical content.

Under the PEPM and WEP, universal proportionality to inertial mass establishes gravitational mass itself as object-intrinsic, removing the positional degree of freedom left open by its localized operational definition. The position index $L$ on $m_L^G$ may therefore be dropped. With the force unit chosen so that $c=1$, inertial and gravitational mass coincide numerically.

The universal equivalence of inertial and gravitational mass bears on their physical nature: either they are distinct manifestations of the same underlying property, or they arise from physically distinct but nomologically related properties. On either account, their equivalence carries ontological content.

\section{Newton's Second Law as Structural Necessity}
Expressing the formal equivalence of inertial and gravitational mass through their operational definitions yields Newton's second law.

\begin{proposition}[Newton's Second Law] Assume the PEPM and WEP, with quantities operationally defined and standards chosen as above. For every body $B$ and every nonzero resultant $F$,
\begin{equation}\label{eq:n2l}
	F = m^G(B) a(B,F).
\end{equation}
\end{proposition}
\begin{proof} 
	Immediate from definition~\ref{def:definertialmass} and proposition~\ref{prop:massequivalence}.
\end{proof}

\medskip
\noindent Mass equivalence and Newton's second law therefore express the same physical content in different forms: formally as $m^G=m^I$, and operationally as $F=m^Ga$.

\section{Consequences}
The resulting structure bears directly on theory, metrology, and experiment. 

\medskip
\noindent
\textbf{Constraints on modified-inertia MOND\@.}  The derivation extends to modified-inertia formulations satisfying symmetry of motion and ordinary initial-value uniqueness---a condition raised within the program itself \parencite{Milgrom2022}. With trajectory history $H$ carried parametrically, the principle of equal heights survives, licensing Galileo's sine law, and the Stevin--Galileo synthesis yields $F=m^I(B,H)\,a(B,H)$ on the incline. Under uniqueness, the body's acceleration at each point of the incline is independent of its history. Standard modified-inertia formulations confine history dependence to the inertial response, leaving the gravitational sector history-independent, so the force, the parallel weight component $\gamma_L m_L^G(B)\sin(\theta)$, carries no history either. History dependence confined to $m^I$ is therefore untenable. The Stevin--Galileo synthesis further renders $F/a$ object-intrinsic. Hence, inertia is independent of both history and acceleration. Within this class of formulations, the MOND rotation-curve signature thus disappears, and with it the galactic phenomenology modified inertia was introduced to recover. Unless symmetry of motion or ordinary initial-value uniqueness is abandoned, the galactic rotation-curve anomaly, under the PEPM, must instead arise from the source or gravitational sector, chiefly through dark matter, modified gravity, or emergent gravity. This answers, under the stated conditions, the question Milgrom poses: modified inertia can retain ordinary initial-value uniqueness or the MOND rotation-curve signature, but not both.

\medskip
\noindent
\textbf{Dimensional structure of mechanics.} The dimensional structure of a physical domain is determined by how many independent dimensions there are and by which quantities serve as base dimensions. Both have been treated as conventional (for a critical discussion, see \cite{Jacobs2024}). The two questions are not on a par. A change of basis leaves the number of independent dimensions untouched; reducing that number requires a relation among quantities previously treated as dimensionally independent---in mechanics, that relation is Newton's second law. The dimensional count is therefore secure only insofar as that law is, and the law is ordinarily posited. That law is here derived from exclusion and invariance principles that presuppose no dimensional relation between force and mass. The dimensional reduction $[F]=MLT^{-2}$ is therefore not inherited from a posited law. It is instead licensed by principles that themselves impose no dimensional structure at all.

\medskip
\noindent
\textbf{Kibble-balance realization of the SI kilogram.} Under the PEPM--WEP constraint, the structural necessity of Newton's second law grounds the location-invariance of the Kibble-balance realization of the SI kilogram.

\medskip
\noindent
\textbf{Unification of WEP tests.} Under the PEPM, universal free fall and mass equivalence are two formulations of the same physical content. Spaceborne free-fall experiments such as MICROSCOPE and terrestrial torsion-balance experiments such as E{\"{o}}t-Wash therefore probe that content through different observables---differential acceleration and differential torque. Their remaining distinction is practical: the cost of space access versus terrestrial noise control.
	
\section{Conclusion}
Newton's second law is a structural necessity of mechanics constrained by the PEPM and WEP\@. Under the PEPM, Suppes' framework operationally grounds gravitational mass and force as independent primitives, structurally coupled by weight in its dual role. The Stevin--Galileo synthesis establishes $F/a$ as object-intrinsic, thereby grounding the operational definition of inertial mass. When constrained by universal free fall, the resulting structure compels the equivalence of inertial and gravitational mass. Newton's second law is mass equivalence expressed in operational terms.

Among direct implications, the PEPM, even without the WEP, constrains modified-inertia MOND and lets free-fall and torsion-balance experiments test the WEP with equal epistemic directness. With the WEP, the full structure further grounds the dimensional reduction of mechanics and the location-invariance of the Kibble-balance realization of the SI kilogram.

More fundamentally, Newton's second law derives its necessity from exclusion and invariance principles whose scope extends beyond mechanics itself.

\begin{refcontext}[sorting=nyt]
	\printbibliography[title={References}]

@book{newton1687,
	author		= "Newton, Isaac",
	title		= "Philosophi{\ae} Naturalis Principia Mathematica",
	year		= "1687",
	publisher	= "Jussu Societatis Regiae ac Typis Josephi Streater",
	address		= "London",
	addendum	= "Translated by Andrew Motte as \textit{The Mathematical Principles of Natural Philosophy} (London: Benjamin Motte, 1729)"
}

@article{mana2022,
	author      = "Mana, Giovanni and Schlamminger, Stephan",
	title       = "The Kilogram: Inertial or Gravitational Mass?",
	journal		= "Metrologia",
	volume      = "59",
	number      = "4",
	pages       = "043001",
	year  		= "2022",
	doi      	= "10.1088/1681-7575/ac7ca7"
}

@article{milgrom1983,
	author		= "Milgrom, Mordehai",
	title		= "A Modification of the {N}ewtonian Dynamics as a Possible Alternative to the Hidden Mass Hypothesis",
	journal		= "Astrophysical Journal",
	volume		= "270",
	number		= "2",
	pages		= "365--370",
	year		= "1983",
	doi			= "10.1086/161130"
}

@article{wagner2012,
	author  = "Wagner, Todd A. and Schlamminger, Stephan and Gundlach, Jens H. and Adelberger, Eric G.",
	title   = "Torsion-Balance Tests of the Weak Equivalence Principle",
	journal = "Classical and Quantum Gravity",
	volume  = "29",
	number  = "18",
	pages   = "184002",
	year    = "2012",
	doi     = "10.1088/0264-9381/29/18/184002"
}

@article{dicasola2015,
	author  = "Di Casola, Eolo and Liberati, Stefano and Sonego, Sebastiano",
	title   = "Nonequivalence of Equivalence Principles",
	journal = "American Journal of Physics",
	volume  = "83",
	number  = "1",
	pages   = "39--46",
	year    = "2015",
	doi     = "10.1119/1.4895342"
}

@book{euler1736,
	author 		= "Euler, Leonhard",
	title 		= "Mechanica sive motus scientia analytice exposita",
	volume 		= "1",
	address 	= "St. Petersburg",
	publisher 	= "Typographia Academiae Scientiarum",
	year 		= "1736"
}

@article{darrigol2020,
	author		= "Darrigol, Olivier",
	title		= "Deducing {N}ewton's Second Law from Relativity Principles: A Forgotten History",
	journal		= "Archive for History of Exact Sciences",
	volume		= "74",
	number		= "1",
	pages		= "1--43",
	year		= "2020",
	doi 		= "10.1007/s00407-019-00230-2"
}

@book{kant1786,
	author    	= "Kant, Immanuel",
	title     	= "Metaphysische Anfangsgr{\"u}nde der Naturwissenschaft",
	publisher 	= "Johann Friedrich Hartknoch",
	address  	= "Riga",
	year     	= "1786",
	addendum   	= "Translated by Michael Friedman as \textit{Metaphysical Foundations of Natural Science} (Cambridge: Cambridge University Press, 2004)"
}

@book{mach1883,
	author		= "Mach, Ernst",
	title		= "Die Mechanik in {i}hrer Entwickelung: Historisch-kritisch {d}argestellt",
	year		= "1883",
	publisher	= "F. A. Brockhaus",
	address		= "Leipzig",
	addendum   	= "Translated by Thomas J. McCormack as \textit{The Science of Mechanics: A Critical and Historical Exposition of Its Principles} (Chicago: Open Court, 1893)"
}

@book{poincare1902,
	author		= "Poincar{\'e}, Henri",
	title		= "La Science et l'Hypoth{\`e}se",
	year		= "1902",
	publisher	= "Flammarion",
	address		= "Paris",
	addendum   	= "Translated by William J. Greenstreet as \textit{Science and Hypothesis} (London: Walter Scott Publishing, 1905)"
}

@book{bridgman1927,
	author		= "Bridgman, Percy  W.",
	title		= "The Logic of Modern Physics",
	year		= "1927",
	publisher	= "Macmillan",
	address		= "New York"
}

@book{jammer2000,
	author    	= "Jammer, Max",
	title     	= "Concepts of Mass in Contemporary Physics and Philosophy",
	publisher 	= "Princeton University Press",
	address   	= "Princeton, NJ",
	year      	= "2000",
	doi 		= "10.1515/9781400823789"
}

@article{saintvenant1845,
	author  	= "Saint-Venant, Adh{\'e}mar J. D.",
	title   	= "M{\'e}moire sur les sommes et les diff{\'e}rences g{\'e}om{\'e}triques et sur leur usage pour simplifier la m{\'e}canique",
	journal 	= "Comptes rendus de l'Acad{\'e}mie des sciences",
	volume 		= "21",
	pages   	= "620--625",
	year    	= "1845"
}

@book{weyl1949,
	author    	= "Weyl, Hermann",
	title     	= "Philosophy of Mathematics and Natural Science",
	publisher 	= "Princeton University Press",
	address   	= "Princeton, NJ",
	year     	= "1949"
}

@article{suppes1951,
	author		= "Suppes, Patrick",
	title		= "A Set of Independent Axioms for Extensive Quantities",
	journal		= "Portugaliae Mathematica",
	volume		= "10",
	number		= "4",
	pages		= "163--172",
	year		= "1951"
}

@article{mckinsey1953,
	author		= "McKinsey, John C. C. and Sugar, Alvin C. and Suppes, Patrick",
	title		= "Axiomatic Foundations of Classical Particle Mechanics",
	journal		= "Journal of Rational Mechanics and Analysis",
	volume		= "2",
	number		= "2",
	pages		= "253--272",
	year		= "1953",
	doi 		= "10.1512/iumj.1953.2.52012"	
}

@book{suppes1957,
	author    	= "Suppes, Patrick",
	title     	= "Introduction to Logic",
	publisher 	= "Van Nostrand",
	address   	= "Princeton, NJ",
	year      	= "1957"
	
}

@incollection{noll1959,
	author   	= "Noll, Walter",
	title     	= "The Foundations of Classical Mechanics in the Light of Recent Advances in Continuum Mechanics",
	booktitle 	= "The Axiomatic Method, with Special Reference to Geometry and Physics",
	editor    	= "Henkin, Leon and Suppes, Patrick and Tarski, Alfred",
	series    	= "Studies in Logic and the Foundations of Mathematics",
	publisher 	= "North-Holland",
	address   	= "Amsterdam",
	year      	= "1959",
	pages     	= "266--281",
	doi 		= "10.1016/S0049-237X(09)70033-3"
}

@book{sneed1971,
	author    	= "Sneed, Joseph D.",
	title     	= "The Logical Structure of Mathematical Physics",
	publisher 	= "D. Reidel Publishing Company",
	address   	= "Dordrecht",
	year      	= "1971",
	doi 		= "10.1007/978-94-010-3066-3"
}

@book{sklar2013,
	author    	= "Sklar, Lawrence",
	title     	= "Philosophy and the Foundations of Dynamics",
	publisher 	= "Cambridge University Press",
	address   	= "Cambridge",
	year      	= "2013",
	doi 		= "10.1017/CBO9781139034340"
}

@article{fox2016,
	author		= "Fox, Craig W.",
	title		= "The {N}ewtonian Equivalence Principle: How the Relativity of Acceleration Led {N}ewton to the Equivalence of Inertial and Gravitational Mass",
	journal		= "Philosophy of Science",
	volume		= "83",
	number		= "5",
	pages		= "1027--1038",
	year		= "2016",
	doi 		= "10.1086/687935"
}

@book{stevin1586,
	author		= "Stevin, Simon",
	title		= "De Beghinselen der Weeghconst",
	year		= "1586",
	publisher	= "Christoffel Plantijn",
	address		= "Leiden",
	addendum   	= "Translated by C. Dikshoorn as \textit{The Elements of the Art of Weighing}, in \textit{The Principal Works of Simon Stevin}, vol. 1, \textit{General Introduction; Mechanics}, edited by E. J. Dijksterhuis, 47--285 (Amsterdam: Swets {\&} Zeitlinger, 1955)"
}

@book{galileo1638,
	author		= "Galilei, Galileo",
	title		= "Discorsi e dimostrazioni matematiche intorno a due nuove scienze",
	year		= "1638",
	publisher	= "Elzevir",
	address		= "Leiden",
	addendum   	= "Translated by Henry Crew and Alfonso de Salvio as \textit{Dialogues Concerning Two New Sciences} (New York: Macmillan, 1914)"
}

@book{duhem1905,
	author    	= "Duhem, Pierre",
	title     	= "Les origines de la statique",
	volumes   	= "2",
	location  	= "Paris",
	publisher 	= "A. Hermann",
	date      	= "1905/1906",
	sortyear	= "1905",
	addendum  	= "Translated by Grant F. Leneaux, Victor N. Vagliente, and Guy H. Wagener as \textit{The Origins of Statics: The Sources of Physical Theory} (Dordrecht: Kluwer Academic Publishers, 1991)"
}

@book{mach1872,
	author    	= "Mach, Ernst",
	title     	= "{D}ie {G}eschichte und die {W}urzel des {S}atzes von der {E}rhaltung der {A}rbeit",
	publisher 	= "J. G. Calve",
	address   	= "Prague",
	year      	= "1872",
	addendum   	= "Translated by Philip E. B. Jourdain as \textit{History and Root of the Principle of the Conservation of Energy} (Chicago: Open Court, 1911)"
}

@book{banks2003,
	author		= "Banks, Erik C.",
	title		= "Ernst Mach's World Elements: A Study in Natural Philosophy",
	year		= "2003",
	publisher	= "Kluwer Academic Publishers",
	address		= "Dordrecht",
	doi 		= "10.1007/978-94-017-0175-4"
}

@book{helmholtz1847,
	author    	= "von Helmholtz, Hermann",
	options		= {useprefix=false},
	sortname	= {Helmholtz, Hermann},
	title     	= "{\"{U}}ber die {E}rhaltung der {K}raft: {E}ine physikalische {A}bhandlung",
	publisher 	= "G. Reimer",
	address   	= "Berlin",
	year      	= "1847",
	addendum  	= "Translated by John Tyndall as \textit{On the Conservation of Force; A Physical Memoir}, in \textit{Scientific Memoirs, Selected from the Transactions of Foreign Academies of Science, and from Foreign Journals: Natural Philosophy}, edited by John Tyndall and William Francis, 114--162 (London: Taylor and Francis, 1853)"
}

@incollection{bevilacqua1994,
	author    	= "Bevilacqua, Fabio",
	title     	= "Helmholtz's {U}eber die {E}rhaltung der {K}raft: The Emergence of a Theoretical Physicist",
	editor    	= "Cahan, David",
	booktitle 	= "Hermann von Helmholtz and the Foundations of Nineteenth-Century Science",
	pages     	= "291--333",
	publisher 	= "University of California Press",
	address   	= "Berkeley",
	year      	= "1994",
	doi			= "10.1525/9780520914094-013"
}

@article{einstein1907,
	author		= "Einstein, Albert",
	title		= "{\"{U}}ber das {R}elativit{\"{a}}tsprinzip und die aus demselben gezogenen {F}olgerungen",
	journal		= "Jahrbuch der Radioaktivit{\"{a}}t und Elektronik",
	volume		= "4",
	pages		= "411--462",
	year		= "1907",	
	addendum   	= "Translated by Anna Beck as \enquote{On the Relativity Principle and the Conclusions Drawn from It,} in \textit{The Collected Papers of Albert Einstein}, vol. 2, \textit{The Swiss Years: Writings, 1900--1909}, English translation supplement (Princeton, NJ: Princeton University Press, 1989)"
}

@article{Milgrom2022,
	author  = "Milgrom, Mordehai",
	title   = "Models of a Modified-Inertia Formulation of {MOND}",
	journal = "Physical Review~D",
	volume  = "106",
	number  = "6",
	pages   = "064060",
	year    = "2022",
	doi     = "10.1103/PhysRevD.106.064060"
}

@article{Jacobs2024,
	author  = "Jacobs, Caspar",
	title   = "In Defence of Dimensions",
	journal = "British Journal for the Philosophy of Science\textup{, ahead of print}",
	year    = "2024",
	doi     = "10.1086/729749"
}
\end{refcontext}

\end{document}